\documentclass[aps,prx,onecolumn,superscriptaddress,notitlepage]{revtex4-2}

\usepackage{soul}
\usepackage{relsize}
\usepackage{lmodern}
\usepackage{slantsc}
\usepackage{bbm}
\usepackage{graphicx}
\usepackage{amsmath}
\usepackage{amsfonts}
\usepackage{amssymb}
\usepackage{amsthm}
\usepackage{amsbsy}
\usepackage{mathrsfs}
\usepackage{varioref}
\usepackage{dsfont}
\usepackage{bm}
\usepackage{color}
\usepackage[usenames,dvipsnames]{xcolor}
\definecolor{myblue}{rgb}{0.2,0.2,0.8}
\definecolor{myblack}{rgb}{0,0,0}
\definecolor{myurl}{rgb}{0.1,0.1,0.4}
\usepackage[colorlinks=true,citecolor=myblue,linkcolor=myblack,urlcolor=myurl]{hyperref}
\usepackage{cleveref}
\usepackage[font = footnotesize, labelfont = bf, justification = RaggedRight]{caption}
\usepackage{subfigure}
\usepackage{booktabs} 
\usepackage{array} 
\usepackage{paralist} 
\usepackage{verbatim} 
\edef\restoreparindent{\parindent=\the\parindent\relax}
\usepackage[parfill]{parskip}
\usepackage{color}
\usepackage[ansinew]{inputenc}
\usepackage[T1]{fontenc}
\newtheorem{theorem}{Theorem}[section] 
\newtheorem{corollary}[theorem]{Corollary}

\newtheorem{definition}[theorem]{Definition}

\newtheorem{remark}[theorem]{Remark}

\begin{document}

\title{Incompatibility of effects in general probabilistic models}
\author{Roberto Beneduci}
\affiliation{Dipartimento di Fisica, Universit\`a della Calabria, Italy}
\affiliation{INFN gruppo collegato  Cosenza, Italy}
\author{Leon  Loveridge} \thanks{Paul Busch was an original coauthor of this paper, with Beneduci. Later, Neil Stevens also joined as a coauthor. This, then unpublished note, has since been referred to in, e.g., \cite{stevens2014steering} and \cite{Busch2016a}. The work of Stevens has since been removed, and Loveridge has made some additions and modifications. Since we believe Paul would have wanted to see a final version, we have opted instead to acknowledge his contribution here.}
\affiliation{Quantum Technology Group, Department of Science and Industry Systems, University of South-Eastern Norway, 3616 Kongsberg, Norway}


\begin{abstract}
We give a necessary and sufficient condition for the incompatibility of a pair of effects in a general probabilistic model in which the state space is a total convex space, which can be obtained by minimising a real parameter. This has an interpretation in terms of the least noise that must be included to make the given pair compatible.
\end{abstract}

\maketitle

\begin{center}
    {\bf Dedicated to Stan Gudder on the occasion of his 85'th birthday.}
\end{center}

\section{Introduction}
In \cite{Wolf}, it was shown that in finite dimensional quantum mechanics, a pair of effect operators $E$ and $F$ are incompatible (alternatively, not jointly measurable) if and only if they enable the violation of the CHSH inequality. The proof is based on the following result: two effects $E$ and $F$ in a finite dimensional Hilbert space $\mathcal{H}$ are incompatible if and only if there is an effect $S$ and a number $\lambda_0>1$, such that $E+F\leq\lambda_0 I+S$. In this paper we show that in the class of general probabilistic models in which the state space is a total convex space, which is considerably more general than the finite spaces often treated in the literature, the incompatibility of a pair of effects can again be cast in terms of the value of a real parameter being greater than one.

The idea of a general probabilistic model stems from the operational, or convex, approach to quantum mechanics \cite{Ludwig,Davies,Holevo,Gudder,Busch,Beltrametti1}. A major aim of this programme is to recover the Hilbert space formulation of quantum mechanics as a particular instance of a more general model. However, it is an active research field in its own right, and has dramatically improved our understanding of key differences between classical physics, quantum physics, and yet more general theories which are neither classical nor quantum (e.g., \cite{Barnum1}). Key insights have come from the stimulus provided by quantum information theory, in which the differing information processing capabilities can be analysed (e.g., \cite{barrett2007information, Barnum}); by now a vast range of phenomena have been investigated (see e.g. \cite{plavala2021general} for a recent introduction).

 There are now numerous investigations of incompatibility in general probabilistic models, which also shed light on the above questions; for instance it is now known that the state space being a simplex is equivalent to all observables being jointly measurable \cite{plavala2016all,kuramochi2020compatibility}, which singles out classical theory as the only theory without incompatibility. The relationship between incompatibility and non-locality has also been developed \cite{jenvcova2018incompatible}, though a direct analogue of the result in \cite{Wolf} remains open.
 

The generalisation of \cite{Wolf} which is presented below characterises the incompatibility of two effects in terms of the value of a real parameter $\lambda_0>1$, realising a key step in understanding the connection to CHSH in this most general setting. The result relies on the observation that the infimum over the set of all $\lambda$ for which a pair of effects are incompatible is in fact a minimum (Theorem 3.4), which is crucial for Corollary 3.5 in which the incompatibility is fully characterised in terms of $\lambda_0$. We show, moreover, that this parameter may be interpreted in terms of the least noise that must be added in order to make the effects compatible.

\section{General probabilistic models}

A general probabilistic model is defined by its set of states, which is considered to be primitive and is described by a convex set $\Omega$.  Convexity is crucial since it assures that a probabilistic mixture of states is still a state. The observables are derived objects, defined as affine maps $F:\Omega\mapsto M_1^+(X)$ from the state space, $\Omega$,  to the space of probability measures on a space $X$ of possible outcomes in a measurement of  $F$. In this section we briefly recall the main ideas in the operational approach, along with a theorem showing how to construct from $\Omega$ a base norm space and, by duality, an order-unit space.

\begin{definition}
The state space $\Omega$ of a general probabilistic model is a convex subset of a real linear vector space $V$. 
\end{definition}
\begin{definition}
An affine functional is a map $f:\Omega\to \mathbb{R}$ such that 
$$f(\lambda_1\omega_1+\dots,\lambda_n\omega_n)=\sum_{i=1}^n\lambda_i\,f(\omega_i),$$
\noindent
 for any finite set of positive numbers, $\lambda_1,\dots,\lambda_n$, $\sum_i\lambda_i=1$. The unit functional $u$ is the functional such that $u(\omega)=1$ for all $\omega\in\Omega$.
\end{definition}
\begin{definition}
An effect is an affine functional $f$ such that $0\leq f(\omega)\leq 1$ for any $\omega\in \Omega$. The space of effects is denoted by $\mathcal{E}(\Omega)$.
\end{definition}
\noindent
In the following, $A^b(\Omega)$ denotes the space of bounded affine functionals.   The order in $A^b(\Omega)$ is given by: $f\leq g$ if and only if $(g-f)(\omega)\geq 0$ for each $\omega\in \Omega$. The positive cone in $A^b(\Omega)$ is the set $C^+=\{f\in A^b(\Omega),\,f(\omega)\geq 0,\,\forall \omega\in \Omega\}$.
\begin{definition}
A convex space $\Omega$ is total convex if $f(\omega_1)=f(\omega_2)$, for any effect $f\in\mathcal{E}(\Omega)$, implies $\omega_1=\omega_2$.
\end{definition}
\noindent
In other words, a total convex space $\Omega$ is such that its elements are separated by the space of effects. Under this hypothesis, the following theorem can be proven.
\begin{theorem}\label{teo1}\cite{Beltrametti}
Let $\Omega$ be a total convex space. Then, the span of $\Omega$, $V(\Omega)$, is a base norm space and $A^b(\Omega)$ is an order unit Banach space.
\end{theorem}
Theorem \ref{teo1} ensures that it is possible to give a topological structure to the space of states $\Omega$, which is a subset of the norm space $V(\Omega)$. Moreover, the theorem implies that the space of effects is a subset of a Banach space, pointing to the connection between the convex approach and the Hilbert space formulation of quantum mechanics. 
\begin{remark}
We recall that $\Omega$ is the base for the positive cone $V(\Omega)^+$ and that each linear functional in the topological dual of $V(\Omega)$, ${V(\Omega)}^*$, is the linear extension of an affine functional in $A^b(\Omega)$. Since (see \cite{Beltrametti}) the order unit of ${V(\Omega)}^*$ is the linear functional on $V(\Omega)$ which takes the value $1$ on every element of $\Omega$, the order unit in $A^b(\Omega)$ is the unit function $u$ of $A^b(\Omega)$. Moreover \cite{Beltrametti}, the order unit norm $\|f\|$ of $f\in {V(\Omega)}^*$ is $\sup\{|f(\omega)|,\,\omega\in Conv(\Omega\cup-\Omega\}=\sup\{|f(\omega)|,\,\omega\in\Omega\}$, which coincides with the sup-norm $\|\cdot\|_{\infty}$ on $A^b(\Omega)$. Finally, there is a canonical embedding $J:\Omega\to{A^b(\Omega)}^*$ (where ${A^b(\Omega)}^*$ is the topological dual of ${A^b(\Omega)}$) which is given by $J(\omega) f=f(\omega)$ for all $\omega\in\Omega$, $f\in A^b(\Omega)$.
\end{remark}

We can introduce the weak*-topology  on $A^b(\Omega)$ as follows: a base of neighborhoods of $f$ is given by the sets 
\begin{equation}\label{weak1}
N(f;\,\omega_1,\dots,\omega_n;\epsilon)=\{g\in A^b(\Omega)\,;\,|f(\omega_i)-g(\omega_i)|\leq\epsilon, \,i=1,\dots,n\}
\end{equation}
\noindent
where, $\omega_1,\dots,\omega_n\in \Omega$, $\epsilon>0$.
In that topology, a sequence of functionals $f_n$ converges to a functional $f$ if $\lim_{n\to\infty}f_n(\omega)=f(\omega)$ for each $\omega\in \Omega$ (point-wise convergence). 

It is worth remarking that the weak*-topology on $A^b(\Omega)$ is induced by the weak*-topology on ${V(\Omega)}^*$, $\sigma({V(\Omega)}^*,V(\Omega))$. Indeed, a base of neighborhoods for $f\in{V(\Omega)}^*$ in the weak*-topology $\sigma({V(\Omega)}^*,V(\Omega))$ is given by
\begin{equation}\label{weak2}
N(f;\,\omega_1,\dots,\omega_n;\epsilon)=\{g\in {V(\Omega)}^*\,:\,|f(\omega_i)-g(\omega_i)|\leq\epsilon,\, i=1,\dots,n\}
\end{equation}
where, $\omega_1,\dots,\omega_n\in V(\Omega)$, $\epsilon>0$. 
\begin{remark}\label{compactness}
The Alaoglu-Banach-Bourbaki theorem assures the compactness of the unit ball $({V(\Omega)}^*)_1=\{f\in {V(\Omega)}^*:\sup_{\omega\in\Omega}|f(\omega)|\leq1\}$ in the weak*-topology (\ref{weak2})  on ${V(\Omega)}^*$. (See \cite{Conway}, page 130). On the other hand (see equations (\ref{weak1}) and (\ref{weak2})), the compactness of $({V(\Omega)}^*)_1$ implies the compactness of the unit ball $(A^b(\Omega))_1=\{f\in A^b(\Omega)\,:\,\|f(\omega)\|_\infty\leq1\}$ in the weak*-topology (\ref{weak1}) on $A^b(\Omega)$. 
\end{remark}

\section{Generalized observables and joint measurability in a general probabilistic model}

\begin{definition}
In a general probabilistic model an observable (measurement) is a map $F:\mathcal{B}(X)\to\mathcal{E}(\Omega)$ where $\mathcal{B}(X)$ is the Borel $\sigma$-algebra of a topological set $X$ and 
   \begin{equation*}
    F\big(\bigcup_{n=1}^{\infty}\Delta_n\big)(\omega)=\sum_{n=1}^{\infty}F(\Delta_n)(\omega),\quad\forall\omega\in\Omega
    \end{equation*}
    \noindent 
 with $\{\Delta_n\}$ a countable family of disjoint sets in $\mathcal{B}(X)$. Moreover, $F$ is normalized, i.e., $F(X)=u$.
\end{definition}
 The set $X$ is interpreted as the set of values that the measurement of the observable described by $F$ can assume. For each $\omega\in\Omega$, $F(\cdot)(\omega)$ is a probability measure on $X$ and  $F(\Delta)(\omega)$ is interpreted as the probability that a measurement of $F$ gives a results in $\Delta$.  Therefore, an observable can also be defined as a map from the state space $\Omega$ to the space of probability measures on $X$, $M_1^+(X)$.
\begin{definition}
If the set $X$ is countable, an observable $F$ is denoted by $F=\{F_x\}_{x\in X}$ and is called discrete. 
\end{definition}
In the case of a discrete observable $F$, the value $F_x(\omega)$ is interpreted as the probability that the result of a measurement of $F$ in the state $\omega$ is $x$.   The simplest nontrivial observable is the dichotomic observable $F=\{f,u-f\}$ which is uniquely defined by the effect $f$. 
\begin{definition}
Two observables $F_1:\mathcal{B}(X_1)\to\mathcal{E}(\Omega)$ and $F_2:\mathcal{B}(X_2)\to\mathcal{E}(\Omega)$ are jointly measurable if there is a third observable $F:\mathcal{B}(X_1\times X_2)\to\mathcal{E}(\Omega)$ such that $F(\Delta_1\times X_2)=F_1(\Delta_1)$ and $F(X_1\times\Delta_2)=F_2(\Delta_2)$.  
\end{definition}
In the case of two dichotomic observables $\{e,u-e\}$ and $\{f,u-f\}$ , the general condition for joint measurability can be cast in the following form (see also \cite{stevens2014steering}):  $\{e,u-e\}$ and $\{f,u-f\}$ are jointly measurable if there is an effect $g$ such that 
\begin{align}\label{joint}
g&\leq e\\\nonumber
g&\leq f\\
e+f&\leq g+u\nonumber
\end{align} 
Since a dichotomic observable is uniquely defined by a single effect, the joint measurability condition (\ref{joint}) can be viewed as a definition for the joint measurability of two effects $e$ and $f$. 

We may now prove the following theorem which generalizes to the infinite dimensional Hilbert space setting some of the results in \cite{Wolf}, and holds also in the more abstract class of general probabilistic models introduced above, to give a necessary and sufficient condition for the incompatibility of two effects. (See Corollary \ref{incompatibility} below).

\begin{theorem}\label{min}
Let $e$ and $f$ be two effects and consider the set $\Lambda:=\{\lambda\in\mathbb{R}\,\vert\,\,\,\exists\, g\leq e,\,f\,;\,g+\lambda u \geq e+f\}$. Then, $\lambda_0:=\inf\Lambda\in\Lambda$.
\end{theorem}

\begin{proof}
By the definition of greatest lower bound, there exists a decreasing sequence $\{\lambda_n\}_{n\in\mathbb{N}}\subset\Lambda$ such that $\lambda_n\to\lambda_0$.

\noindent
For each $n\in\mathbb{N}$, let us define the set $\mathcal{A}_{n}=\{g\in\mathcal{E}(\Omega)\,\vert\, g+\lambda_n u\geq e+f\,;\,g\leq e,f\}$ where $\mathcal{E}(\Omega)$ is the set of effects. Now, we divide the proof into 4 steps.

\noindent
1) \textbf{We prove that} $\mathcal{A}_{n+1}\subset\mathcal{A}_{n}$. 

\noindent
Suppose $g\in\mathcal{A}_{n+1}$. Since $\lambda_n\geq\lambda_{n+1}$ we have,
\begin{equation*}
g+\lambda_{n}u\geq g+\lambda_{n+1}u\geq e+f.
\end{equation*}
\noindent
hence, $g\in\mathcal{A}_{n}$.

\noindent
2) \textbf{We prove that}, for each $n\in\mathbb{N}$, $\mathcal{A}_{n}$ is closed with respect to the weak*- topology. 

\noindent
Let $g$ be an accumulation point for $\mathcal{A}_{n}$. Then, there is a net $\{g_{t}\}\subset\mathcal{A}_{n}$ which converges to $g$, $g_{t}\to g$. Moreover, for each $t$, 
\begin{equation*}
g_t+\lambda_n u\geq e+f.
\end{equation*}
Therefore,
$$g+\lambda_n u\geq e+f$$
\noindent
which implies $g\in\mathcal{A}_{n}$.

\noindent
3) \textbf{We prove that} $\cap_{n=1}^{\infty}\mathcal{A}_{n}\neq\emptyset$.

\noindent
Recall that, for each  $n\in\mathbb{N}$, the set $\mathcal{A}_{n}$ is contained in the unit ball $(A^b(\Omega))_1$ which is compact in the weak*-topology (\ref{weak1}). (See Remark \ref{compactness}). 

\noindent
Moreover the family of sets $\{\mathcal{A}_{n}\}_{n\in\mathbb{N}}$ has the finite intersection property, i.e., each finite subfamily has a nonempty intersection. Therefore (see theorem 1, page 136, in Ref. \cite{Kelly}), by the compactness of $(A^b(\Omega))_1$, the intersection $\cap_{n=1}^{\infty}\mathcal{A}_{n}$ must be nonempty.

\noindent
4) \textbf{We get the thesis}.

\noindent
Since $\cap_{n=1}^{\infty}\,\mathcal{A}_{n}\neq\emptyset$, there is an effect $g$ such that $g\in\mathcal{A}_{n}$ for each $n\in\mathbb{N}$. That implies,
$$g+\lambda_n u\geq e+f,\quad\forall n\in\mathbb{N}.$$
\noindent
Therefore,
$$g+\lambda_0 u\geq e+f.$$   
\end{proof}
\noindent
Theorem \ref{min} ensures that the greatest lower bound of $\{\lambda\in\mathbb{R}\,\vert\,\,\,\exists\, g\leq e,\,f\,;\,g+\lambda u \geq e+f\}$ is a minimum and implies the following characterization of the incompatibility of two effects.

\begin{corollary}\label{incompatibility}
Two effects $e$ and $f$ are incompatible (i.e., not jointly measurable) iff   $\lambda_0=\min\{\lambda\in\mathbb{R}\,\vert\,\,\,\exists\, g\leq e,\,f\,;\,g+\lambda u \geq e+f\}>1$.
\end{corollary}

Note that theorem \ref{min} is essential for the corollary since otherwise $\lambda_0=1$ would not imply the compatibility of $e$ and $f$. Note moreover that $\lambda_0\leq2$ since $g=0$ and $\lambda=2$ satisfy the conditions in the corollary.

Given a pair of incompatible effects $e$ and $f$, the parameter $\sigma_0 :=2(1-\lambda_0^{-1}) \in [0,1]$ can be interpreted as a measure of the least noise, in the sense of a particular smearing, that must be added to $e$ and $f$ to make them compatible.   Denoting $e' = u-e$, an observable $F=\{e,\,e'\}$ can be smeared by means of a Markov kernel $ \{\mu_{11},\mu_{12},\mu_{21},\mu_{22}\,|\, \mu_{11}+\mu_{21}=\mu_{12}+\mu_{22}=1\}$, with $\{\mu_{11}e+\mu_{1
2}e',\,\mu_{21}e+\mu_{22}e'\}$---the smearing of $F$---interpreted as a noisy version of $F$. The Markov kernel provides two probability measures, $\{\mu_{1,1},\mu_{2,1}\}$ and $\{\mu_{2,1},\mu_{2,2}\}$ which describe the noise added to the observable $F$. The smearing $\tilde{e} = k^{-1}e$ and $\tilde{f} = k^{-1}f$; $k>0$ corresponds to the Markov kernel
\begin{equation*}\label{MK1}
    \{\mu_{11}=\frac{1}{k},\,\mu_{12}=0,\,\mu_{21}=1-\frac{1}{k},\,\mu_{22}=1\}.
\end{equation*}
In this interpretation, $k=1$ means no noise while $k\to\infty$ corresponds to the maximum noise. Indeed, $\tilde{e}\to0$ and $\tilde{e}'\to u$ as $k\to\infty$.  

Fixing $\lambda_0 > 1$ as above, we see that the effects $\tilde{e}_0 := e/\lambda_0$ and $\tilde{f}_0 := f/\lambda_0$ are always compatible. In order to show that $\sigma_0$ corresponds to the least noise to be added to $e$ and $f$ to make them compatible, suppose there is $1<k<\lambda_0$ such that $\tilde{e} = k^{-1}e$ and $\tilde{f} = k^{-1}f$ are compatible. Then, 
$$\lambda_0'=\min\{\lambda\in\mathbb{R}\,\vert\,\,\,\exists\, g\leq \tilde{e},\,\tilde{f}\,;\,g+\lambda u \geq \tilde{e} + \tilde{f}\}\leq 1.$$

\noindent
Moreover,

$$\lambda_0'=\min\{\lambda\in\mathbb{R}\,\vert\,\,\,\exists\, kg\leq e,\,f\,;\,kg+k\lambda u \geq e + f\}\leq 1.$$
\noindent
Hence, there is an effect $g$ such that 

$$ kg\leq e,\,f\,\quad\,kg+\lambda_0'k u \geq e + f,\quad \lambda_0'\leq 1.$$

Then there is an effect $g'=kg$ such that 

$$ g'\leq e,\,f\,\quad\,g'+\lambda_0'k u \geq e + f,\quad \lambda_0'\leq 1,\,\,k<\lambda_0.$$

But this contradicts the fact that $\lambda_0=\min\{\lambda\in\mathbb{R}\,\vert\,\,\,\exists\, g\leq e,\,f\,;\,g+\lambda u \geq e + f\}$ since $\lambda_0'k<\lambda_0$, justifying the conclusion that $\sigma_0$ is the minimum noise required (with respect to the given Markov kernel) to make any pair $e$ and $f$ compatible. This smearing differs from that given in  \cite{stevens2014steering, plavala2016all, jenvcova2018incompatible}, in both explicit form and in a more qualitative sense: the smearing given above (and also in \cite{Wolf}) arises from a stochastic Markov kernel, whereas the one appearing in \cite{stevens2014steering, plavala2016all, jenvcova2018incompatible} is doubly stochastic.



\section{Conclusions and outlook}
We have given a necessary and sufficient condition for the compatibility of a pair of effects in a general probabilistic model in which the state space is a total convex space. This constitutes a step towards establishing the full analogue of the result of \cite{Wolf} in this general setting. Our final observation points to connections between (in)compatibility and degree of smearing for Markov kernels more general than those appearing in the literature thus far, and will be the subject of further enquiry.

\section*{Acknowledgment}
\noindent
 RB: The present work has been realized in the framework of the activities of the INDAM (Istituto Nazionale di Alta Matematica). \


\bibliographystyle{apsrev4-2}

\bibliography{bib.bib}

\end{document}